\documentclass[letterpaper]{article}
\usepackage{uai2020}
\usepackage[margin=1in]{geometry}

\usepackage{times}

\usepackage[round]{natbib}
\usepackage{graphicx}
\usepackage{float}
\usepackage{algorithm}
\usepackage{algpseudocode}
\usepackage{subcaption}
\usepackage[colorlinks,citecolor=black]{hyperref}

\usepackage{xcolor}

\newcommand{\obs}{{\textrm{obs}}}
\newcommand{\kn}{{\textrm{kn}}}
\newcommand{\un}{{\textrm{un}}}
\newcommand{\textand}{{\textrm{and}}}
\newcommand{\textor}{{\textrm{or}}}
\newcommand{\bg}{{\textrm{bg}}}

\newcommand{\given}{\mid}
\newcommand{\joint}{{\textrm{joint}}}

\newcommand{\nnodes}{20}

\usepackage{amsmath,amsfonts,amsthm,amssymb}
\usepackage{bbm}
\usepackage{bm}

\newtheorem{defn}{Definition}

\newtheorem{example}{Example}

\newtheorem{thm}{Theorem}

\newtheorem{assumption}{Assumption}
\newtheorem{remark}{Remark}
\newtheorem*{assumption*}{Assumption}

\newcommand\independent{\protect\mathpalette{\protect\independenT}{\perp}}
\def\independenT#1#2{\mathrel{\rlap{$#1#2$}\mkern2mu{#1#2}}}
\newcommand{\kron}{\mathbbm{1}}

\newcommand{\cE}{\mathcal{E}}
\newcommand{\cG}{\mathcal{G}}
\newcommand{\cH}{\mathcal{H}}
\newcommand{\cI}{\mathcal{I}}
\newcommand{\cM}{\mathcal{M}}
\newcommand{\cN}{\mathcal{N}}

\newcommand{\bbP}{\mathbb{P}}

\DeclareMathOperator{\pa}{pa}

\DeclareMathOperator{\ch}{ch}

\DeclareMathOperator{\skel}{skel}
\DeclareMathOperator{\descendants}{de}

\DeclareMathOperator{\pre}{pre}

\newcommand{\rvI}{{\mathsf{I}}}
\title{Permutation-Based Causal Structure Learning \\ with Unknown Intervention Targets}

\author{ {\bf Chandler Squires} \\
LIDS, IDSS \\
MIT\\
\texttt{csquires@mit.edu} \\
\And
{\bf Yuhao Wang}  \\
Statistical Laboratory         \\
University of Cambridge \\
\texttt{yw505@cam.ac.uk}\\
\And
{\bf Caroline Uhler}   \\
LIDS, IDSS \\
MIT    \\
\texttt{cuhler@mit.edu}\\
}

\begin{document}

\maketitle

\begin{abstract}
  We consider the problem of estimating causal DAG models from a mix of observational and interventional data, when the intervention targets are partially or completely unknown. This problem is highly relevant for example in genomics, since gene knockout technologies are known to have off-target effects. We characterize the interventional Markov equivalence class of DAGs that can be identified from interventional data with unknown intervention targets. In addition, we propose a provably consistent algorithm for learning the interventional Markov equivalence class from such data. The proposed algorithm greedily searches over the space of permutations to minimize a novel score function. The algorithm is nonparametric, which is particularly important for applications to genomics, where the relationships between variables are often non-linear and the distribution non-Gaussian. We demonstrate the performance of our algorithm on synthetic and biological datasets. Links to an implementation of our algorithm and to a reproducible code base for our experiments can be found at \href{https://uhlerlab.github.io/causaldag/utigsp}{https://uhlerlab.github.io/causaldag/utigsp}.
\end{abstract}

\section{INTRODUCTION}
Causal models are a prerequisite for answering scientific, sociological, and technological questions across disciplines~\citep{FLN00,Pearl:00,RHB00,SGS00}; examples are ``what genetic activity is responsible for cancer?'' or ``what is the effect on unemployment of raising minimum wage?''. This necessity has generated intense interest in \emph{causal structure learning}, i.e., the problem of learning a causal graphical model that represents the causal relationships of different elements in a complex system from data. Typically, the causal model is in the form of a \emph{directed acyclic graph} (DAG).

Since different causal DAG models can generate the same observational distribution, a DAG is in general only identifiable up to its \textit{Markov equivalence class} (\emph{MEC}) from observational data~\citep{verma1990equivalence}. Interventional data is necessary for reducing the ambiguity. Given observational and interventional data the
identifiability of the underlying causal DAG model improves to a smaller equivalence class known as the $\cI$-MEC~\citep{HB12,YKU18}. With the advent of gene editing technologies in genomics, high-throughput interventional gene expression data is being produced~\citep{perturb-seq}. Therefore, an important problem in this field is to fully utilize such data to infer the finest equivalence class of causal DAGs describing the data. This is made particularly challenging since gene knockout experiments are known to have severe off-target effects, i.e., the CRISPR-Cas gene-editing technology performs cleavage at unknown genome sites other than their intended target~\citep{Fu13,Wang15}. Not accounting for these additional targets while learning causal structure leads to model misspecification, and thus incorrect conclusions. Hence it is critical to develop causal inference methods that can make use of observational and interventional data when the intervention targets are partially or completely unknown. This is the purpose of the present paper.

A variety of methods have been proposed for causal structure learning from observational and interventional data when the intervention targets are known. This includes the algorithms GIES~\citep{HB12} and IGSP~\citep{WSYU17,YKU18} under the assumption of causal sufficiency, i.e., when there are no latent confounders, and ACI~\citep{MCM16}, HEJ~\citep{HEJ14} and COmbINE~\citep{TT15} that allow for latent confounders. Since these algorithms assume that all intervention targets are known a priori, they will in general be inconsistent in the presence of off-target effects, which may misinform downstream decision-making. To make use of interventional data with unknown intervention targets, \citet{EM07} proposed a dynamic programming algorithm. However, it is limited both in terms of scalability and requiring parametric assumptions. A different approach to this problem is given by the \emph{invariant causal inference framework}~\citep{Mein16,RHPM15,GSKZ17}. While this approach comes with consistency guarantees, it makes various assumptions that are unlikely to hold in the context of genomics. In particular, interventions can only affect the distribution of the internal noises of the intervened targets and the functional relationship between each node and its parents is assumed to be linear. Most recently, \citet{mooij2016joint} proposed the \emph{Joint Causal Inference} (\emph{JCI}) framework, which can be used to adapt an existing observational causal inference algorithm into a method for causal structure learning from interventional data with unknown targets. In this paper, we develop a new algorithm for learning from interventional data with unknown targets, and will also compare our algorithm to the JCI framework.

Our main contributions are as follows:

\begin{itemize}
\item We show that under a specific faithfulness assumption, all intervention targets are identifiable. Importantly, this implies that the degree of identifiability of the underlying causal model is the same with unknown intervention targets as when the intervention targets are known. 

\item By introducing a score function that is minimized by graphs in the true $\cI$-MEC, we develop a provably consistent greedy algorithm that simultaneously learns the intervention targets as well as the $\cI$-MEC from a mix of observational and interventional data with unknown intervention targets. 

\item We demonstrate the efficacy of our algorithm on synthetic and biological datasets. 
\end{itemize}

\section{PRELIMINARIES AND RELATED WORK}
\label{sec:preliminaries}

\subsection{Causal DAG model}
 Let  $\cG = ([p], \cE)$ be a directed acyclic graph (DAG) with node set $[p] := \{1, \cdots, p\}$ and edge set $\cE$ representing a causal model where each node $i$ is associated with a random variable $X_i$. Let $f$ denote the density of the data-generating distribution $\bbP$ over the random vector $X := (X_1, \cdots, X_p)$. By the causal Markov property, the density function $f$ is \emph{Markov} with respect to the DAG $\cG$, i.e., the density function $f$ \emph{factorizes} with respect to the DAG $\cG$: 
\begin{align*}
f(x) = \prod_{i \in [p]} f_i(x_i \mid x_{\pa_\cG(i)}),
\end{align*}
where $\pa_\cG(i)$ denotes the set of nodes that are parents of $i$ in the DAG $\cG$. A basic result for DAG models~\citep[Section~3.2.2]{LAU96} is that $\bbP$ is Markov with respect to a DAG $\cG$ if and only if 
the set of conditional independence relations in $\bbP$ is entailed by the set of d-separation statements\footnote{d-separation is reviewed in Supplementary Material~A} in $\cG$, i.e., for any disjoint sets $A$, $B$ and $C$, $X_A$ is conditionally independent from $X_B$ given $X_C$ whenever $A$ is d-separated from $B$ given $C$. The \emph{faithfulness assumption} that is commonly assumed in existing causal inference algorithms is the assertion that the converse is also true, 
i.e., that the set of conditional independence relations in $\bbP$ entail all d-separation statements in $\cG$. The main justification of the faithfulness assumption is that the Lebesgue measure of distributions unfaithful with respect to a DAG $\cG$ is zero~\citep{Pearl:00}.

Let $\cM(\cG)$ denote the set of distributions that are Markov with respect to $\cG$. Two DAGs $\cG_1$ and $\cG_2$ are \emph{Markov equivalent}, denoted $\cG_1 \sim \cG_2$, if $\cM(\cG_1) = \cM(\cG_2)$. \citet{verma1990equivalence} showed that $\cG_1 \sim \cG_2$ if and only if $\cG_1$ and $\cG_2$ have the same skeleton and v-structures. Moreover, if $\cG_1\sim \cG_2$, then $\cG_1$ and $\cG_2$ can be transformed to one another by a sequence of \emph{covered edge reversals}, where we call an edge $i \to j$ in a DAG $\cG$ \emph{covered} if $\pa_{\cG}(j) = \pa_{\cG}(i) \cup \{i\}$. By a slight abuse of notation, we will also use $\cM(\cG)$ to denote the set of DAGs that are Markov equivalent to $\cG$, i.e., the Markov equivalence class of $\cG$.

\subsection{Interventions}
\label{sec_int}

Interventions on random variables can be used to improve the identifiability of the underlying causal model. A theoretical framework for modeling interventions was developed in \citet{eberhardt2007interventions}. 
A \emph{perfect intervention} assumes that all causal dependencies between intervened targets and their causes are removed~\citep{eberhardt2007interventions}. As an example, consider a perfectly performed gene knockout experiment, where the expression of a gene is set to zero and hence all interactions between gene $i$ and its upstream regulators are eliminated.

In practice, interventions often cannot fully remove the causal dependencies between an intervened target and its causes, but rather \emph{modify} their causal relationship~\citep{eberhardt2007interventions}. For example, in genomics, an intervention may only inhibit the expression of a gene~\citep{DLQ16}. Such interventions are known as \emph{imperfect}. The issue of whether or not an intervention is perfect or imperfect is conceptually orthogonal to the issue of whether or not the intervention has unknown targets. For example, a chemical treatment that perfectly prevents the expression of an unknown handful of genes would be an example of a \textit{perfect} intervention with \textit{unknown} targets. On the other hand, injecting a cell with extra copies of mRNA from gene~A would be an example of an \textit{imperfect} intervention with no unknown targets, since the expression of gene~A still depends on the gene regulatory network, which has not been affected. This paper is concerned with the problem of causal structure discovery from interventional data (from perfect or imperfect interventions) with \mbox{\emph{unknown intervention targets}.} 

Let $I\subseteq [p]$ denote a perfect or imperfect intervention target and let $f^\obs$ and $f^I$ denote the densities of the observational (i.e., no interventions) and interventional distributions, respectively. A pair $(f^\obs, f^I)$ is $I$-\emph{Markov} with respect to a DAG $\cG$ if $f^\obs$ and $f^I$ are Markov with respect to $\cG$ and for any non-intervened variable $j \in [p] \setminus I$, it holds  that 
\begin{equation}
    \label{inv_property}
    f^I(x_j \mid x_{\pa_\cG(j)}) = f^\obs (x_j \mid x_{\pa_\cG(j)}),
\end{equation}
i.e., the conditional distributions of the non-intervened variables are invariant across the observational and interventional distributions~\citep{YKU18}. This $I$-Markov property implies the following factorization of the interventional distribution $f^I$ with respect to $\cG$:
 \begin{align}\label{eq:int}
f^I(x) = \prod_{i \not\in I} f^\obs(x_i \mid x_{\pa_\cG(i)}) \prod_{i \in I} f^I(x_i \mid x_{\pa_\cG(i)}).
\end{align}
Let $\cM_I(\cG)$ denote the set of distributions $I$-Markov with respect to $\cG$. Then, as in the non-interventional setting, two DAGs $\cG_1$ and $\cG_2$ are in the same $I$-\emph{Markov equivalence class}, if $\cM_I(\cG_1) = \cM_I(\cG_2)$~\citep{HB12, YKU18}.

\subsection{Causal structure discovery algorithms}
\label{subsec:background-algorithms}

Causal inference algorithms can largely be categorized into three approaches, namely \emph{constraint-based} methods, \emph{score-based} methods, and their hybrids. Constraint-based methods, including the prominent PC algorithm~\citep{SGS00}, learn the causal model by treating causal inference as a constraint satisfaction problem and estimate the underlying Markov equivalence class by a sequence of conditional independence tests. Score-based methods, such as GES~\citep{Meek1997} and its interventional adaptation GIES~\citep{HB12}, assign a score to each Markov equivalence class and learn the Markov equivalence class of the data-generating DAG by greedily optimizing a penalized likelihood score. In addition, hybrid algorithms such as GSP~\citep{SWMU17} and its interventional adaptation IGSP~\citep{WSYU17,YKU18} have been proposed that construct a score function based on conditional independence tests. All these algorithms assume either that there is no interventional data or that the intervention targets are known. The main contributions of this paper is to provide a consistent causal inference algorithm in the setting where the intervention targets are unknown.

Recent work \citep{mooij2016joint} introduces a framework for causal structure learning using data from heterogeneous ``contexts", including data from different interventions, to which we will limit our discussion. The \emph{joint causal inference} (\emph{JCI}) framework associates with intervention $I_k$ a binary random variable $\rvI_k$, with $\rvI_k = 1$ denoting that the data comes from the distribution $k$. The vector $\rvI$ has at most a single non-zero entry (i.e., $|\rvI|_0 \leq 1$), and $\rvI = \mathbf{0}$ denotes that the data comes from $f^\obs$. Thus, the joint distribution of the \textit{system variables} $x$ and the \textit{intervention variables} $\rvI$ is
$$
f^\joint(x, \rvI) = f^\obs(x)^{\kron_{\rvI = \mathbf{0}}} \prod_{k=1}^K f^k(x)^{\kron_{\rvI_k = 1}}.
$$
This distribution can be represented by the \emph{JCI-DAG}, denoted $\cG^\joint_*$, which fuses the true underlying causal DAG $\cG^*$ with a complete graph over the intervention variables, and adds the edge $\rvI_k \rightarrow x_i$ if $i \in I_k$.


To apply JCI to a causal structure learning algorithm, the algorithm must be capable of incorporating the following assumptions as background information:
\begin{itemize}
    \item \textbf{``Exogeneity"}: System variables do not cause intervention variables.
    \item \textbf{``Generic context"}: The intervention variables are fully connected.
\end{itemize}

The JCI framework has been applied to a variety of constraint-based and scored-based methods, but has not been applied to any hybrid methods. In Section \ref{sec:algorithm}, we provide an adaptation of GSP that can incorporate the background information required for JCI, leading to a new algorithm, \textit{JCI-GSP}. Then, we show that the performance of JCI-GSP suffers from treating intervention variables equivalently to system variables, and propose an improved algorithm, \textit{Unknown-Target IGSP (UT-IGSP)} to overcome this problem.


\section{IDENTIFIABILITY WITH UNKNOWN INTERVENTION TARGETS}
\label{sec:identifiability}

In order to define consistency of a causal inference algorithm in the setting where the intervention targets are unknown, we first need to characterize the interventional Markov equivalence class in this setting. In the following, we first briefly review the graphical characterization of the interventional Markov equivalence class when all intervention targets are known and then show that the equivalence class is the same even in the setting where the intervention targets are unknown. This means that the degree of identifiability of the underlying causal DAG model is unchanged whether the intervention targets are known or unknown. 

\subsection{Preliminaries}
\label{sec:characterization}
We consider the setting where we have data from $K$ interventional experiments. Let $I^k$ denote the intervention targets of experiment $k$ and let $f^k$ denote the corresponding interventional distribution. We denote the full list of intervention targets by $\cI = (I^1, \cdots, I^K)$. Notice that we assume throughout that we also have access to purely observational data. This 
assumption is satisfied in most experimental designs in practice.

The $I$-Markov property in Section~\ref{sec_int} can easily be extended to the setting of multiple interventional experiments by replacing the invariance property (\ref{inv_property}) by
\begin{align*}
f^k(x_i \mid x_{\pa_\cG(i)}) = f^{k'}(x_i \mid x_{\pa_\cG(i)})
\end{align*}
for all $k, k' \in [K]$ and all random variables $X_i$ where $i \not\in I$ and $i \not\in I'$ (see also \citet{YKU18}). We denote the resulting $\cI$-Markov equivalence class with respect to a DAG $\cG$ by $\cM_\cI(\cG)$ and the equivalence relation by $\sim_\cI$.

A graphical characterization of $\cI$-Markov equivalence was provided by \citet{YKU18}. Let  $\cG^\cI$ denote the DAG $\cG$ along with additional $\cI$-vertices $\{\zeta_k \}_{k \in [K]}$ and $\cI$-edges $\{\zeta_k \rightarrow i \}_{i \in I^k, I^k \in \cI}$ (this is known as the \emph{interventional DAG}, or \emph{$\cI$-DAG}; a concrete example is provided in Figure~\ref{fig:I-DAG}). Then $\cG_1 \sim_\cI \cG_2$ if and only if $\cG_1^\cI$ and $\cG_2^\cI$ have the same skeleton and v-structures. Similarly as in the non-interventional setting, the $\cI$-Markov property connects the $\cI$-DAG to invariance of conditional distributions via d-separation.
Specifically, if $\{f^k\}_{k \in [K]}$ is $\cI$-Markov with respect to $\cG$, then for disjoint $A$ and $C$, $f^I(x_A \mid x_C) = f^\obs(x_A \mid x_C)$ whenever $A$ and $\zeta_I$ are d-separated given $C \cup \zeta_{\cI \setminus I}$, denoted as $(A \independent \zeta_I \mid C \cup \zeta_{\cI \setminus I})_{\cG^\cI}$. 
The $\cI$-Markov equivalence class of a graph $\cG$ can be represented by a partially directed graph, the $\cI$-essential graph, which has a directed edge $i \rightarrow j$ in the $\cI$-essential graph if the edge $i \to j$ is oriented in the same direction for every DAG in $\cM_\cI(\cG)$, and has an undirected edge $i - j$ if the edge is oriented in different directions for DAGs in $\cM_\cI(\cG)$.

\begin{figure}[t!]
\centering
\includegraphics[width=.3\textwidth]{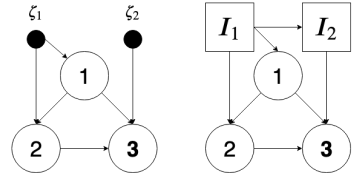}
\vspace{-.2cm}
\caption{The $\cI$-DAG (left) $\cG^\cI$ and JCI-DAG (right) $\cG^\joint$ for a complete DAG and the interventions $I_1 = \{1,2\}$ and $I_2 = \{3\}$.} \label{fig:I-DAG}
\end{figure}

\subsection{Main results}
\label{sec:extension}

Let the estimated set of intervention targets be
$$
\hat{I}^k = \{ x_i \in [p] \mid f^k(x_i \mid x_S) \neq f^\obs(x_i \mid x_S) ~\forall~ S \subseteq [p] \setminus \{i\} \}.
$$
By definition, $f^k(x_i \mid x_{\pa_\cG(i)}) = f^\obs(x_i \mid x_{\pa_\cG(i)})$ for $i \not\in I^k$, so we always have $\hat{I}^k \subseteq I^k$. The following assumption ensures that $\hat{I}^k = I^k$.

\begin{assumption}[Direct $\cI$-faithfulness]\label{ass:identifiability}
Given an interventional distribution $f^k$ with targets $I^k$, we assume that $f^k(x_i \mid x_S) \neq f^\obs(x_i \mid x_S)$ for any node $i \in I^k$ and any subset $S \subseteq [p] \setminus \{i\}$.
\end{assumption}
This assumption rules out situations in which node $i$ has been intervened on, but there is some set $S$ for which the conditional distribution $f^k(x_i \mid x_S)$ is unaffected. Note that this is equivalent to adjacency-faithfulness between intervention variables and their children in the JCI-DAG.

Assumption~\ref{ass:identifiability} is not required by known-target interventional causal inference algorithms (see for example~\citet{TP01,YKU18})\footnote{We show in Supplementary Material \ref{appendix:incomparability} that Assumption~\ref{ass:identifiability} is incomparable to the assumptions in \citet{YKU18}}. Thus, it is of interest to understand whether Assumption~\ref{ass:identifiability} is truly necessary for causal inference in the setting with unknown intervention targets.  We end this section with the following example showing that when Assumption~\ref{ass:identifiability} is violated, the underlying $\cI$-Markov equivalence class may not be identifiable, i.e., Assumption~\ref{ass:identifiability} is necessary for any causal inference algorithm in the setting where the intervention targets are unknown.

\begin{example}[Necessity of Assumption~\ref{ass:identifiability}] 
Let $f$ be Markov to the DAG $1 \rightarrow 2$ and $I_1 = \{2\}$, with $f^\obs(x_1) = \cN(0, 1)$, $f^\obs(x_2 \mid x_1) = \cN(x_1, 1)$, and $f^1(x_2 \mid x_1) = \cN(0.5x_1, 1.75)$. We have $f^\obs(x_2) = f^1(x_2) = \cN(0, 2)$, violating Assumption \ref{ass:identifiability}. The DAG $2 \rightarrow 1$ with intervention set $I'_1 = \{1\}$ and distributions $g^\obs(x_2) = \cN(0, 2)$, $g^\obs(x_1 \mid x_2) = \cN(0.5x_2, 0.5)$ and $g^1(x_1 \mid x_2) = \cN(0.25x_2, 0.875)$ gives the same set of interventional distributions, so one cannot distinguish between the two DAGs despite the fact that they are in different interventional Markov equivalence classes.
\end{example}

\section{ALGORITHM AND ITS CONSISTENCY}
\label{sec:algorithm}

In Section~\ref{sec:identifiability}, we have shown that the full list of intervention targets is identifiable and hence the underlying $\cI$-MEC is the same as in the setting where all intervention targets are known. One approach for learning the $\cI$-MEC is to first estimate $\{\hat{I}^k\}_{k \in [K]}$, and then apply algorithms such as IGSP~\citep{YKU18} that operates in the setting where the intervention targets are known. However, estimating $I^k$ directly may require an exhaustive search over all $2^{p-1}$ subsets of variables, which is intractable for real-world applications with hundreds or thousands of nodes.

In the following, we provide a greedy algorithm that learns the $\cI$-MEC as well as a complete list of intervention targets \textit{simultaneously}. Importantly, we show that this greedy algorithm is consistent, i.e., it outputs the correct $\cI$-MEC with increasing sample size.

\subsection{Preliminaries}
\label{subsec:alg-prelims}

The proposed algorithm is an interventional adaptation of the greedy sparsest permutation (GSP) algorithm \citep{SWMU17} that was proposed for causal inference in the purely observational setting. GSP is a permutation-based causal inference algorithm that associates a score to each permutation $\pi$, i.e., an ordering of the random variables $X_1, \cdots, X_p$. It then greedily moves between permutations to optimize the given score function. More precisely, each permutation $\pi$ is associated to its minimal I-MAP, i.e., the DAG $\cG_\pi := ([p], \cE_\pi)$ given by:
$$
i \rightarrow j \in \cE_\pi \;\Longleftrightarrow\; i <_\pi j \;\textrm{and}
 \quad i \not\independent j \mid \pre_\pi(i, j) \setminus \{i, j\}.
$$

Where $\pre_\pi(i, j)$ denotes all nodes coming before either $i$ or $j$ in the permutation $\pi$. From any starting permutation $\pi_0$, GSP uses a depth-first-search approach to find a new permutation $\tau$, where the moves between permutations are defined by covered edge reversals. If there exists $\tau$ obtained by a covered edge reversals such that the number of edges in the minimal I-MAP $\cG_\tau$ is strictly smaller than the number of edges in $\cG_{\pi_0}$, i.e., $|\cG_\tau| < | \cG_{\pi_0}|$, then $\pi_0$ is set to $\tau$ and the search continues. Otherwise, $\cM(\cG_{\pi_0})$ is returned. GSP is consistent under the faithfulness assumption, i.e., it outputs the correct Markov equivalence class in the purely observational setting~\citep{SWMU17}.



\subsection{Main results}

Just like GSP, the proposed algorithm uses a greedy search in the space of permutations to determine the data-generating $\cI$-MEC. Instead of using the number of edges in the minimal I-MAPs as the scoring function, we introduce a new scoring function that can make use of the interventional data without requiring knowledge of the intervention targets.

We consider the following setting: We are given the distributions $f^\obs, f^1, f^2, \cdots, f^K$ based on the intervention targets $\cI := \{I^1, I^2, \cdots, I^K\}$, which are partially known or completely unknown. For each experiment we denote any known intervention targets by $I_\kn^k \subseteq I^k$. Given a permutation $\pi$ and the corresponding minimal IMAP $G_\pi$, we may estimate targets of interventions $k$ as follows:
$$
\cI_\pi^k = I^k_\kn \cup \{i \mid f^\obs(x_i \mid x_{\pa_{G_\pi}(i)}) \neq f^k(x_i \mid x_{\pa_{G_\pi}(i))} \}
$$

and assign the following score function:
\begin{align*}
S(\pi) := |\cG_\pi| + \sum_{k=1}^K |\cI_\pi^k|.
\end{align*}
Here, $|\cG_\pi|$ corresponds to the number of edges in $\cG_\pi$.

To provide some intuition for the  two summands in $S(\pi)$: 
The first summand $|\cG_\pi|$ restricts the global optimum to be in the correct (observational and thus larger) MEC, while the second summand is used to further restrict the global optimum to be within the correct (interventional and thus smaller) $\cI$-MEC. In the finite sample regime, the first summand is estimated by performing conditional independence tests using samples from just the observational distribution. The second summand is estimated by performing conditional invariance tests. In the Gaussian case, this corresponds to testing equality of regression coefficients and conditional variances, as detailed in Supplementary Material~\ref{appendix:conditional-invariance-testing}. In the nonparametric setting, conditional invariance tests can be performed by a combination of nonparametric regression and testing for the equality of the residual distributions, as discussed in \citet{HPM18}. The next remark provides intuition for how the second summand in the score function can pin down the correct $\cI$-Markov equivalence class.

\begin{remark}[Intuition for the score function]
Consider an interventional distribution with intervention targets $I^k \subseteq [p]$ based on the causal DAG $\cG_{\pi^*}$. Under direct $\cI$-faithfulness, $I_\pi^k \supseteq I^k$, so $S(\pi)$ is minimized if we can find $\pa_{G_\pi}(j)$ such that $f^k(x_j \mid x_{\pa_{\cG_\pi}(j)}) = f^\obs(x_j \mid x_{\pa_{\cG_\pi}(j)})$ for all $j \not\in I^k$, in which case $I_\pi^k = I^k$. For example, this invariance will hold if $\pa_{\cG_\pi}(j) = \pa_{\cG^*}(j)$. However, if $f^k$ is $\cI$-faithful to $G^\cI$ and there is some $j \not\in I^k$ such that $j$ is d-connected to $\zeta_{I_k}$ given $\pa_{\cG_\pi}(j) \cup \zeta_{\cI \setminus I_k}$, then $S(\pi)$ will not be minimized. In other words, minimizing the second summand may orient edges and hence increase identifiability of the underlying DAG model.
\end{remark}

The interventional data is not only used to increase the degree of identifiability of the underlying causal model, but also to restrict the search directions in our greedy search algorithm. This is achieved by introducing a more restrictive version of a covered edge.

\begin{defn}
Given a partially unknown intervention set $\cI := \{I^1, \cdots, I^K\}$, an arrow $i \rightarrow j$ in the minimal I-MAP $\cG_\pi$ is \emph{$\cI$-covered} if it is a covered arrow in $\cG_\pi$ and for all $k$ such that $i \in I^k_\kn$, it holds that $f^k(x_j \given x_{\pa_{\cG_\pi}(j)}) \neq f^\obs(x_j \given x_{\pa_{\cG_\pi}(j)})$.
\end{defn}


\begin{algorithm}[tb]
	\caption{Unknown-target IGSP (UT-IGSP)}
	\label{alg:igsp}
	\begin{algorithmic}
		\State{\textbf{Input:} Distributions $f^\obs, f^1, \cdots, f^K$ and partially known intervention sets $\cI^\kn := \{I_\kn^1, I_\kn^2 \cdots, I_\kn^K\}$, a starting permutation $\pi_0$.}
		\State{\textbf{Output:} A permutation $\pi$ and associated minimal I-MAP $\cG_{\pi}$, a complete set of estimated intervention targets $\cI := \{I^1_\pi, \cdots, I^K_\pi \}$.}
		\State 1. Set $\pi := \pi_0$;
		\State 2. Using a depth-first search with root $\pi$, search for a permutation $\tau$ such that $S(\tau) < S(\pi)$ and that the corresponding minimal I-MAP $\cG_\tau$ is connected to $\cG_\pi$ by a list of $\cI$-covered arrow reversals. If such $\tau$ exists, set $\pi$ as $\tau$ and continue this step; otherwise, return $\pi$, $\cG_\pi$ and $\cI := \{I^1_\pi, \cdots, I^K_\pi \}$.
	\end{algorithmic}
\end{algorithm}

Our proposed algorithm for causal structure discovery from interventional data with unknown or partially known intervention targets is provided in Algorithm~\ref{alg:igsp} (which we name \emph{UT-IGSP} for \emph{Unknown Target Interventional Greedy Sparsest Permutation} Algorithm). Next, we prove consistency of this algorithm under the following assumption.
\begin{assumption} [$\cI$-faithfulness assumption]\label{assumption:faithfulness}
	Let $\cI$ be a list of intervention targets. The set of distributions $\{f^\obs\} \cup \{f^I\}_{I_\cI}$ is \emph{$\cI$-faithful} with respect to a DAG $\cG$ if $f^\obs$ is faithful with respect to $\cG$ and for any $I^k \in \cI$ and disjoint $A, C \subseteq [p]$, we have that $(A \independent \zeta_k \given C \cup \zeta_{[K] \setminus \{k\}})_{\cG^\cI}$ if and only if $f^k(x_A \given x_C) = f^\obs(x_A \given x_C)$.
\end{assumption}

Under the $\cI$-Markov property it holds that $(A \independent \zeta_I \given C \cup \zeta_{\cI \setminus I})_{\cG^\cI}$ implies $f^I(x_A \given x_C) = f^\obs(x_A \given x_C)$. As in the purely observational setting, Assumption~\ref{assumption:faithfulness} gives the assertion that the converse is true. Note that the $\cI$-faithfulness assumption is stronger than Assumption~\ref{ass:identifiability}, but in either case, the set of distributions violating the assumption is degenerate\footnote{Formally, for a linear Gaussian model with Gaussian interventional distributions, the set of parameters violating the assumption has Lebesgue measure zero.}, just as for the faithfulness assumption. Similar faithfulness assumptions have been made in prior work on learning from interventional data with known targets, in particular, Assumption 4.4 and Assumption 4.5 in \citet{YKU18}. Since our algorithm must also learn the intervention targets, it is not surprising that Assumption~\ref{assumption:faithfulness} implies both of these assumptions as special case.

Next we show that UT-IGSP (Algorithm~\ref{alg:igsp}) is consistent under Assumption \ref{assumption:faithfulness}. While a direct proof can be obtained and was developed in a preprint of this work, a simpler proof is now given using the JCI framework of~\cite{mooij2016joint}. In Supplementary Material \ref{appendix:background-knowledge}, we also show that GSP is easily capable of handling the exogeneity and generic context assumptions described in Section \ref{subsec:background-algorithms} without any impact on its consistency guarantees. Hence the JCI framework can be applied to GSP, giving rise to JCI-GSP, which is described in Supplementary Material \ref{app:jci-gsp}. 
Compared to UT-IGSP, JCI-GSP uses estimated intervention targets in its definition of covered edges. As discussed in Remark \ref{remark:jci-vs-utigsp-edges} and Example \ref{ex:jci-gsp-inconsistent} below, this leads JCI-GSP to be more sensitive to faithfulness violations than UT-IGSP. Consistent with this observation, UT-IGSP achieves superior performance on synthetic data as shown in Section \ref{subsec:simulated-data}. This motivates our introduction of UT-IGSP as the main algorithm in this paper and the use of JCI-GSP as a proof tool.


\begin{thm}\label{thm:main}
Under Assumption \ref{assumption:faithfulness}, UT-IGSP (Algorithm~\ref{alg:igsp}) and JCI-GSP  are consistent in discovering the $\cI$-Markov equivalence class of the data-generating DAG $\cG_{\pi^*}$ as well as the set of interventional targets in each interventional distribution.
\end{thm}
\begin{proof}
It suffices to establish that JCI-GSP is consistent, since at each minimal I-MAP, Algorithm~\ref{alg:igsp} has a superset of the search directions that JCI-GSP does. To establish consistency of JCI-GSP (i.e., that there is a weakly decreasing sequence from every minimal I-MAP of $f^\joint$ to $\cG_{\pi^*}$), it suffices to show that every minimal I-MAP $G^\joint_\pi$ is a minimal I-MAP of $f^\joint$.

In the construction of $G^\joint_\pi$, we may partition the CI tests into three types:
\begin{enumerate}
\vspace{-0.4cm}
    \item between two intervention variables;
    \vspace{-0.2cm}
    \item between an intervention variable and a system variable;
    \vspace{-0.2cm}
    \item between two system variables;
    \vspace{-0.4cm}
\end{enumerate}
The first type of CI test is handled by the background knowledge that the intervention variables are pairwise adjacent.
Since all intervention variables are before system variables, all CI tests between intervention variables and system variables are of the form $\rvI_k \independent x_i \mid x_C, \rvI_{[K] \setminus \{k\}}$. This CI statement is equivalent to the invariance statement $f^k(x_i \mid x_C) = f^\obs(x_i \mid x_C)$, so every CI test of the second type is consistent by the $\cI$-faithfulness assumption. Finally, every CI test of the third type is consistent by the faithfulness assumption on $f^\obs$, which completes the proof.
\end{proof}

\begin{remark}\label{remark:jci-vs-utigsp-edges}
Note that the $\cI$-covered edges of $\cG_\pi^\cI$ are a superset of the covered edges in $\cG_\pi^\joint$. For $i \rightarrow j$ to be covered in $\cG_\pi^\joint$, we must have $j \in \ch_{\cG_\pi^\joint}(\zeta_k)$ for all $k$ such that $i \in I^k_\kn$, i.e. $j \in I^k$. Then, by the definition of an intervention, $f^k(x_j \mid x_{\pa_{\cG_\pi}(j)}) \neq f^\obs(x_j \mid x_{\pa_{G_\pi}(j)})$, so $i \rightarrow j$ is also $\cI$-covered in $G_\pi^\cI$. Thus, UT-IGSP always has at least as many search directions as JCI-GSP. Moreover, UT-IGSP may have strictly more search directions than JCI-GSP, and thus UT-IGSP is consistent under strictly weaker conditions than $\cI$-faithfulness and strictly weaker conditions than required for JCI-GSP. This is demonstrated in the following example.
\end{remark}


\begin{example}\label{ex:jci-gsp-inconsistent}
    Let $\cG = \{ 1 \rightarrow 3, 2 \rightarrow 3 \}$, $I^1 = \{1\}$, and $I^1_\kn = \emptyset$. Suppose $f^\joint$ is faithful to $G^\joint_*$ (equivalently in this case, $(f^k)_{k \in [K]}$ is $\cI$-faithful to $G^\cI$) except that $f^1(x_3) = f^\obs(x_3)$. Then in JCI-GSP, there are no reversible covered edges in $\cG^\joint_{312}$, so JCI-GSP is not consistent. However, UT-IGSP is consistent, since it may reverse $1 \rightarrow 3$ to get to $\cG^\joint_{132}$, then $3 \rightarrow 2$ to get to $G^\joint_*$, as shown in Figure~\ref{fig:jci-gsp-counterexample}.
\end{example}

Our definition of $\cI$-covered edges also differs from the definition of $\cI$-covered edges in IGSP (for known intervention targets). In IGSP, a covered edge $i \rightarrow j$ is considered $\cI$-covered if the marginals of $x_j$ are invariant, i.e., $f^k(x_j) = f^\obs(x_j)$ for $k$ such that $i \in I^k_\kn$. The IGSP definition immediately leads to problems in settings with unknown targets, since this condition is violated if $j \in I^k$. Furthermore, the definitions differ even in settings with no unknown targets. In both algorithms, false negatives when determining $\cI$-covered edges are problematic, since the path to the sparsest I-MAP may be cut off. Our definition, unlike the IGSP definition, adapts to the strength of the interventions, leading to less false negatives when interventions have enough power.

\begin{figure}[t!]
    \centering
    \includegraphics[width=.5\textwidth]{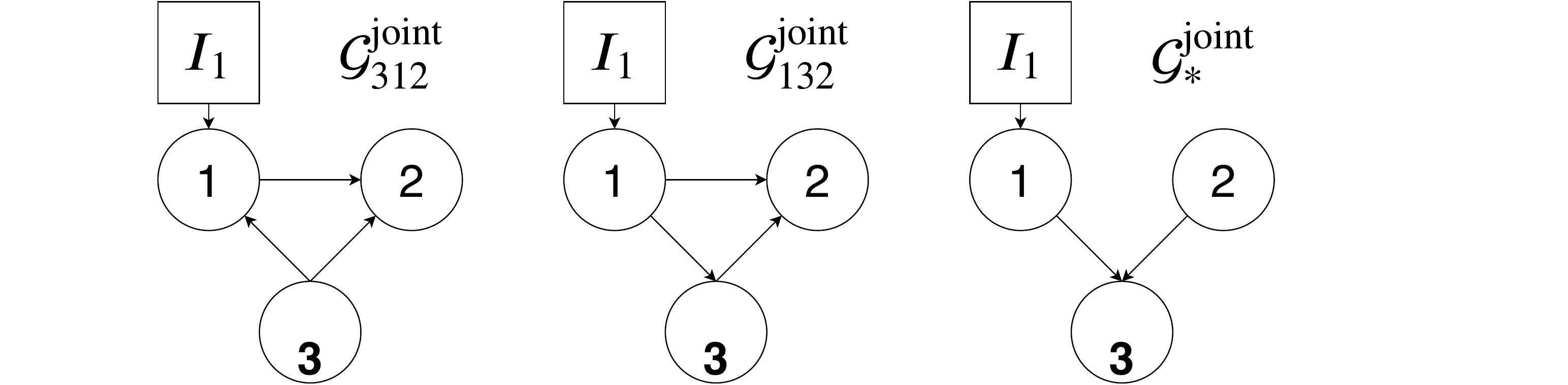}
    \caption{Minimal I-MAPs $\cG_{312}^\joint$, $\cG_{132}^\joint$, and $G_*^\joint$ from Example \ref{ex:jci-gsp-inconsistent}, where UT-IGSP is consistent but JCI-GSP is not.}
    \label{fig:jci-gsp-counterexample}
\end{figure}

In the following section, we will show that UT-IGSP outperforms JCI-GSP, which suggests that the consistency of UT-IGSP under weaker faithfulness conditions has an effect in the finite-sample case. We will also show that it outperforms IGSP even in settings without off-target effects, which suggests that our definition of $\cI$-covered edges is preferable even in the known-target setting.





\section{EMPIRICAL RESULTS}

\subsection{Simulated data}\label{subsec:simulated-data}

\begin{figure*}[t]
	\centering
	\begin{subfigure}[t]{\textwidth}
	    \includegraphics[width=\textwidth]{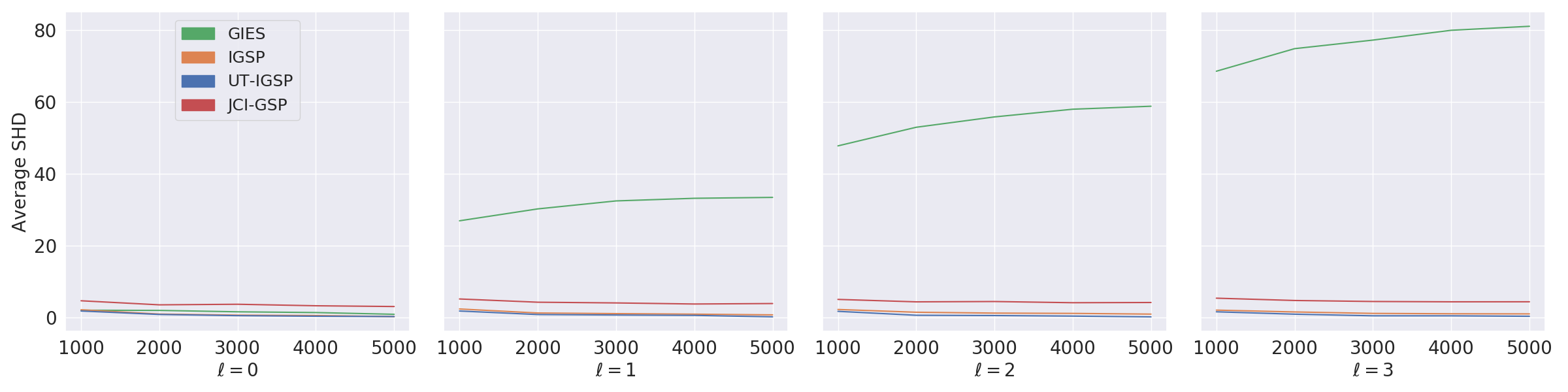}
	    \caption{Average structural Hamming distance from the true $\cI$-essential graph}
	\end{subfigure}
	\begin{subfigure}[t]{\textwidth}
	    \includegraphics[width=\textwidth]{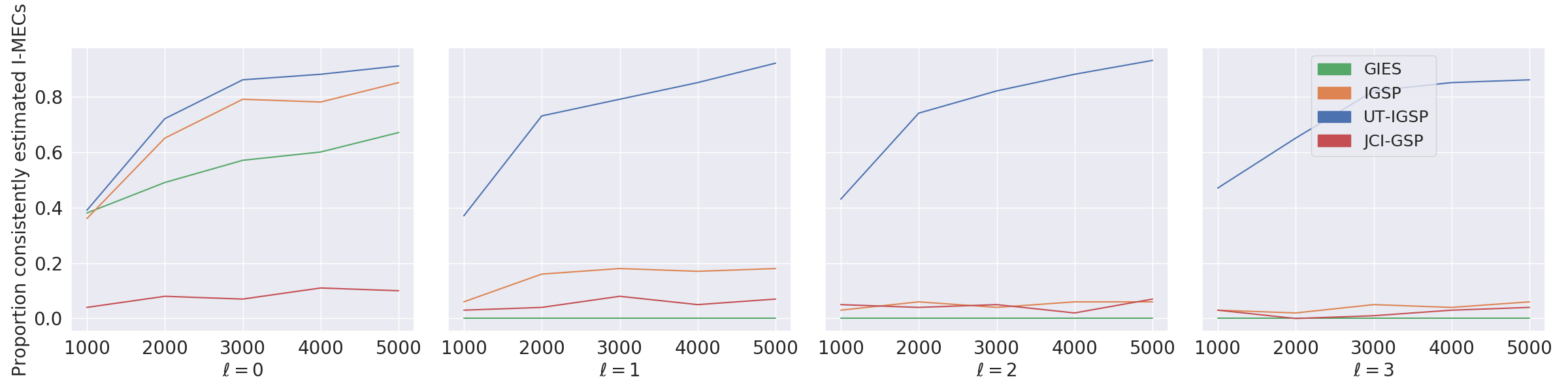}
	    \caption{Proportion of correctly estimated $\cI$-MECs}
	\end{subfigure}
	
	\caption{Performance of different methods as a function of number of samples and number of off-target effects ($\ell$) for 100 Gaussian DAG models on \nnodes~nodes. (a) corresponds the average Hamming distance between the estimated $\cI$-essential graph and the true $\cI$-essential graph, (b) corresponds to the proportions of consistently estimated $\cI$-MECs within the 100 randomly generated Gaussian DAG models.} \label{fig:simulations}
\end{figure*}

In this section, we compare 
UT-IGSP with prior algorithms that assume known intervention targets, namely GIES \citep{HB12} and IGSP \citep{YKU18}, and also JCI-GSP which handles unknown intervention targets, on the task of determining the $\cI$-MEC from interventional data with partially known targets. In this simulation study, we consider data from a linear structural equation model with Gaussian noise, i.e.
\begin{align*}
X = W^T X + \epsilon,
\end{align*}
where the matrix $W$ is upper-triangular  with $W_{ij} \neq 0$ if and only if $i \to j \in \cG_{\pi^*}$ and $\epsilon \sim \cN(0, I_p)$. For each simulation setting, we generated 100 realizations of Erd\"os-R\'enyi DAGs with expected neighborhood size $s=$1.5 and $p = \nnodes$ nodes. To each edge we assigned a  weight $W_{ij}$ sampled independently at random from the uniform distribution on $[-1, -.25] \cup [.25, 1]$, 
ensuring that the edge weights are bounded away from zero. For each DAG, we generated a list of intervention targets $\cI =\{I^1, \cdots, I^5\}$. We first generated known intervention targets $I^1_\kn, \cdots, I^5_\kn$ by randomly picking $5$ nodes from the node set $[p]$ without replacement and assigning one intervention to each of $I^1_\kn, \cdots, I^5_\kn$. Then we generated each set of unknown intervention targets $I^k_\un$ by picking $\ell = 0, \cdots, 3$ nodes from the set $[p] \setminus I^k_\kn$.
Given target $I_k = I_k^\kn \cup I_k^\un$, we generated the interventional distribution via the \emph{shift intervention} model. More precisely, for each node $i \in I_k$, we change its internal noise variance $\epsilon_i$ from mean 0 to mean 1. The shift in mean makes for a simple, easy-to-understand setting, and in the genetic setting, can be thought of as resulting from a gene overexpression experiment.
In each study, we compared different algorithms for $n$ samples from each interventional distribution with $n = 1000,~2000, \ldots,~5000$. 

In each simulation, we ran GIES with its default parameters from the package \texttt{pcalg}. For UT-IGSP and IGSP, we chose a significance level of $\alpha = 10^{-5}$. 

\begin{figure*}
	\centering
	\begin{subfigure}[t]{.45\textwidth}
	\includegraphics[width=\textwidth]{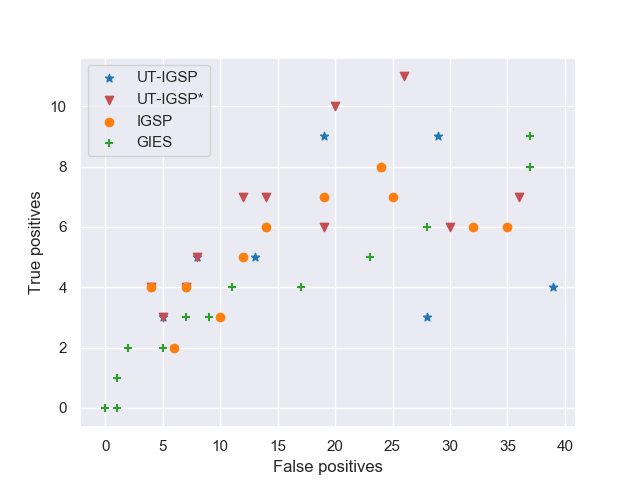}
	    \caption{Directed edge recovery}
	\end{subfigure}
	~
	\begin{subfigure}[t]{.45\textwidth}
	\includegraphics[width=\textwidth]{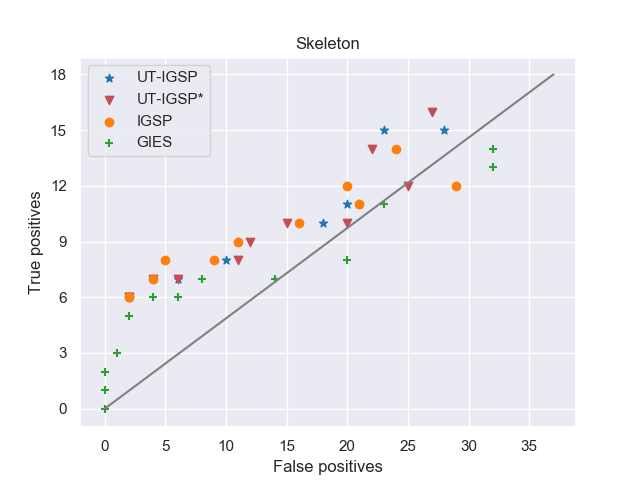}
	    \caption{Skeleton recovery}
	\end{subfigure}
	\caption{ROC curves for models estimated by GIES, IGSP, and UT-IGSP. UT-IGSP* indicates the results of running UT-IGSP with no intervention targets specified. The solid line corresponds to 
	random guessing.} \label{fig:sachs}
	\vspace{-0.1cm}
\end{figure*}

Figure~\ref{fig:simulations} shows the structural Hamming distance\footnote{Given two partially directed graphs, the SHD measures the minimum number of edge additions/deletions/conversions between directed and undirected to convert one graph to the other. Therefore, larger SHD means worse performance.} (SHD) of the causal graphs estimated by each algorithm as well as the proportion of consistently estimated $\cI$-MECs as a function of number of samples for the 4 methods. For GIES and IGSP, the $\cI$-essential graph is with respect to known intervention targets, while for UT-IGSP, the $\cI$-essential graph is with respect to known and estimated intervention targets. As expected, when off-target effects exist, UT-IGSP outperforms all other methods. Even with no off-target effects ($\ell = 0$), UT-IGSP outperforms the other methods, suggesting that our definition of $\cI$-covered edges combines well with sparsity-based search. JCI-GSP, although consistent as $n\to\infty$ (Theorem~\ref{thm:main}), performs poorly across all regimes. Analyzing particular cases suggests that this is due to the definition of covered edges in JCI-GSP, which allows the estimated intervention targets to drastically restrict the search space. When the conditional invariance test experiences false negatives (e.g.~due to finite sample size), JCI-GSP tends to cut off paths to the true DAG. Notably, the performance of GIES degrades drastically with increasing off-target effects. In contrast, the performance of IGSP degrades only slightly, suggesting that this method is more robust to the influence of off-target effects. 
Results for the task of intervention target recovery and for perfect interventions are provided in Supplementary Material \ref{appendix:additional-evaluation}. Finally, we note that UT-IGSP scales well on sparse graphs: the average runtime for the 20-node graphs considered here is below 1 second per graph, and is only 20 seconds for $p = 100$, $\ell = 3$, and $s = 1.5$.

\subsection{Biological data}

We evaluated Algorithm \ref{alg:igsp} on a protein mass spectroscopy dataset acquired from cells from the human immune system \citep{sachs2005causal}. The dataset contains 7466 samples measuring the abundance of phosphoproteins and phospholipids under different experimental conditions. These conditions are generated by inhibiting or activating different proteins in the protein signalling network as well as receptor enzymes via various reagents. This allows us to treat data collected from different experiments as data generated from different interventional distributions. Since some of the interventional experiments intervened on both receptor enzymes and signalling proteins and some experiments intervened only on enzymes, in this study, we 
define the observational dataset as the experiment for which only the receptor enzymes were perturbed, while the other 8 interventional datasets correspond to experiments where the signaling molecules have also been perturbed, as described previously in \cite{WSYU17}. This division gives 1755 observational samples and 4091 interventional samples. A conventionally accepted ground-truth network is reported in \citet{sachs2005causal}.

In Figure~\ref{fig:sachs}, we plot the ROC curves of UT-IGSP, IGSP, and GIES for the true DAG and its skeleton. As expected, both IGSP and UT-IGSP outperform GIES in discovering the skeleton as well as directed edges, since they are both nonparametric approaches that allow for non-linear functional relationships. On the other hand, the performance of IGSP and UT-IGSP is comparable. This indicates that the protein signalling data collected by~\citet{sachs2005causal} may not contain off-target effects; consistent with the fact that these experiments were carefully designed to avoid off-target effects. In this setting, UT-IGSP does not have an advantage over the IGSP algorithm. Finally, we ran UT-IGSP without any intervention targets specified, denoted as UT-IGSP*. We found that the performance of UT-IGSP and UT-IGSP* is similar, suggesting that our algorithm may be useful in applications where off-target effects are expected. 

\section{DISCUSSION}
In this paper, we presented a new algorithm with theoretical consistency guarantees to learn the interventional Markov equivalence class in the presence of off-target effects. 
We showed that the $\cI$-Markov equivalence class is identifiable even without prior knowledge of the intervention targets, a theoretical result of independent interest~\citep{EM07,Mein16,RHPM15,GSKZ17}. The application of our method to the analysis of protein signaling data suggests that it is a viable tool for biological data analysis. 

Our method is of relevance beyond the analysis of interventional data in genomics. For example, our method can be used to learn causal graphs when data is generated from heterogeneous observational sources collected from naturally perturbed systems, since we can take each source as an interventional distribution with imperfect interventions and unknown intervention targets. Examples include gene expression data from normal and diseased states or stock data before and after a financial crisis. 
In the future, it would interesting from a theoretical and practical perspective to extend UT-IGSP to handle latent confounding and to apply UT-IGSP to other data sets.

\subsubsection*{Acknowledgements}
Chandler Squires was supported by an NSF Graduate Research Fellowship and an MIT Presidential Fellowship. Caroline Uhler was partially supported by NSF (DMS-1651995), ONR (N00014-17-1-2147 and N00014-18-1-2765), IBM, a Sloan Fellowship and a Simons Investigator Award. We thank the reviewers of an early version of this paper for pointing out the connection of our algorithm to Joint Causal Inference ~\citep{mooij2016joint}, which we used to obtain simplified proofs of our results.

\bibliography{bib}

\clearpage
\appendix
\newcommand{\snum}{S}
\renewcommand{\theequation}{\snum.\arabic{equation}}

{\Large\textbf{Supplementary Material}}

\section{DAG models}\label{appendix:dag-models}
{\bf d-separation.} For a triple of nodes $(i,j,k)$ in a graph $\cG$ such that $i \to k \leftarrow j$, we call $k$ a \emph{collider}. Given a DAG $\cG$, we say that the two nodes $i$ and $j$ are d-connected given a set of nodes $S$ if there exists a directed path that connects $i$ and $j$ such that every non-collider on the path is not in $S$ and that for every collider $k$ on the path, we have that either $k \in S$ or some descendant of $k$ is in $S$. Given disjoint subsets $A$, $B$, and $C$, we say $A$ and $B$ are d-connected given $C$, denoted by $(A \not\independent B \given C)_\cG$, if there exists a d-connecting path given $C$ between any $a \in A$ and $b \in B$. Otherwise, we say $A$ and $B$ are \textit{d-separated}, denoted $(A \independent B \given C)_\cG$.

{\bf Independence map.} For two DAGs $\cG$ and $\cH$, if the set of distributions Markov with respect to $\cG$ is a subset of the distributions Markov with respect to $\cH$, i.e., $\cM(\cG) \subseteq \cM(\cH)$, we call $\cH$ an \textit{independence map} of the DAG $\cG$, denoted as $\cG \leq \cH$. Based on the Markov property we can also conclude that if $\cG \leq \cH$, the set of conditional independence relations entailed by $\cH$ is a subset of the conditional independence relations entailed by $\cG$.

\section{Assumption \ref{ass:identifiability} Incomparability}\label{appendix:incomparability}
We reproduce the assumptions of \citet{YKU18} here:
\begin{assumption*}[4.4 of \citet{YKU18}]
    Let $I^k \in \cI$ with $i \in I^k$. Then $f^k(x_j) \neq f^\obs(x_j)$ for all descendants $j$ of $i$.
\end{assumption*}
\begin{assumption*}[4.5 of \citet{YKU18}]
    Let $I^k \in \cI$ with $i \in I^k$. Then $f^k(x_j \mid x_S) \neq f^\obs(x_j \mid x_S)$ for any child $j$ of $i$ s.t. $j \not\in I^k$ and for all $S \subseteq \ne_{\cG^*}(j) \setminus \{ i \}$, where $\ne_{\cG^*}(j)$ denotes the neighbors of node $j$ in $\cG^*$/
\end{assumption*}
Example \ref{ex:jci-gsp-inconsistent} satisfies Assumption \ref{ass:identifiability}. It does not satisfy Assumption 4.4, since $3$ is a descendant of $1$ but $f^1(x_3) = f^\obs(x_3)$. It does not satisfy Assumption 4.5, since $3$ is a child of $1$, $S = \emptyset$ is a subset of the neighbors of $3$, and again $f^1(x_3) = f^\obs(x_3)$. Thus, Assumption \ref{ass:identifiability} does not imply either Assumption 4.4 or 4.5.

Let $\cG = \{ 1 \rightarrow 2 \}$ and $I^1 = \{2\}$. Let $f^1(x_2) = f^\obs(x_2)$. Then $f$ satisfies Assumption 4.4 and 4.5, since $2$ has no children/descendants, but it does not satisfy Assumption \ref{ass:identifiability}. Thus, Assumption 4.4 and 4.5 do not imply Assumption~\ref{ass:identifiability}.

\section{Conditional Invariance Testing}\label{appendix:conditional-invariance-testing}

For a multivariate Gaussian distribution $f^\obs$, all conditional distributions are also Gaussian, with mean given by a linear combination of the variables in the conditioning set, i.e., $X_i \mid X_C \sim \cN(\beta_{i\mid C} X_C + b_{i\mid C}, \sigma_{i \mid C}^2)$. Thus, two conditional distributions $f^1$ and $f^2$ are the same if and only if the regression coefficients are the same ($H_c$: $\beta^1_{i\mid C} = \beta^2_{i \mid C}$ and $b^1_{i \mid C} = b^2_{i \mid C}$) and the variance are the same ($H_v$: $\sigma^1_{i \mid C} = \sigma^2_{i \mid C}$). By applying Bonferroni correction, to test the null hypothesis $f^1 = f^2$ at significance level $\alpha$, we may test $H_c$ and $H_v$ both at significance level $\frac{\alpha}{2}$. Both $H_c$ and $H_v$ have well-known exact tests; the Chow test and F test, respectively.

\section{Background Knowledge in GSP}\label{appendix:background-knowledge}

\begin{algorithm}[!b]
	\caption{GSP}
	\label{alg:gsp}
	\begin{algorithmic}
		\State{\textbf{Input:} Distribution $f$ and starting permutation $\pi_0$.}
		\State{\textbf{Output:} A permutation $\pi$ and associated minimal I-MAP $\cG_{\pi}$.}
		\State 1. Set $\pi := \pi_0$;
		\State 2. Using a depth-first search with root $\pi$, search for a permutation $\tau$ such that $S(\tau) < S(\pi)$ and that the corresponding minimal I-MAP $\cG_\tau$ is connected to $\cG_\pi$ by a list of $\cI$-covered arrow reversals. If such $\tau$ exists, set $\pi$ as $\tau$ and continue this step; otherwise, return $\pi$, $\cG_\pi$.
	\end{algorithmic}
\end{algorithm}

\subsection{Consistency of GSP}\label{appendix:gsp-proof}
Algorithm \ref{alg:gsp} describes the Greedy Sparsest Permutation (GSP) algorithm when there is no background information. The algorithm was originally introduced and proven to be consistent in \citet{SWMU17}. We now review a simplified proof, so that we have a reference point for proving consistency after adding background knowledge.

Given a DAG $\cG$ and an IMAP $\cH$ of $\cG$, a \emph{Chickering sequence} from $\cG$ to $\cH$ is a sequence of DAGs
$\cG = \cG_0, \cG_1, \ldots, \cG_{M-1}, \cG_M = \cH$ such that $\cG_i$ is an IMAP of $\cG_{i-1}$ and $\cG_i$ is obtained from $\cG_{i-1}$ by either the addition of an edge or a covered edge reversal. \citet{Chickering02} proved the existence of a Chickering sequence between a DAG $\cG$ and any IMAP $\cH$ of $\cG$ by repeated application of the \textsc{ApplyEdgeOperation} algorithm, reproduced in Algorithm \ref{alg:apply-edge-op}. 

\begin{algorithm*}[!t]
	\caption{\textsc{ApplyEdgeOperation}}
	\label{alg:apply-edge-op}
	\begin{algorithmic}
		\State \textbf{Input:} DAGs $\cG$ and $\cH$ where $\cG \leq \cH$ and $\cG\neq \cH$.
        \State \textbf{Output:} A DAG $\cG^\prime$ satisfying $\cG^\prime\leq \cH$ that is given by reversing an edge in $\cG$ or adding an edge to $\cG$.
        \State 1. Set $\cG^\prime := \cG$.
        \State 2. While $\cG$ and $\cH$ contain a node $Y$ that is a sink in both DAGs and for which $\pa_\cG(Y) = \pa_\cH(Y)$, remove $Y$ and all incident edges from both DAGs.
        \State 3. Let $Y$ be any sink node in $\cH$. 
        \State 4. If $Y$ has no children in $\cG$, then let $X$ be any parent of $Y$ in $\cH$ that is not a parent of $Y$ in $\cG$.  
        	Add the edge $X\rightarrow Y$ to $\cG^\prime$ and return $\cG^\prime$. 
        \State 5. Let $D\in\descendants_\cG(Y)$ denote the (unique) maximal element from $\descendants_\cG(Y)$ within $\cH$.  
        	Let $Z$ be any maximal child of $Y$ in $\cG$ such that $D$ is a descendant of $Z$ in $\cG$.  
        \State 6. If $Y\rightarrow Z$ is covered in $\cG$, reverse $Y\rightarrow Z$ in $\cG^\prime$ and return  $\cG^\prime$.
        \State 7. If there exists a node $X$ that is a parent of $Y$ but not a parent of $Z$ in $\cG$, then add $X\rightarrow Z$ to $\cG^\prime$ and return $\cG^\prime$.
        \State 8. Let $X$ be any parent of $Z$ that is not a parent of $Y$. Add $X\rightarrow Y$ to $\cG^\prime$ and return $\cG^\prime$.  
	\end{algorithmic}
\end{algorithm*}

To show that GSP is consistent, we note that $S(\pi) = |\cG_\pi|$ reaches its minimum only if $\cG_\pi \in \cM(\cG^*)$, as shown in \cite{SWMU17}. Thus, it suffices to show that from any $\pi_0$ s.t. $\cG_{\pi_0} \not\in \cM(\cG^*)$, there is some $\pi_1$ connected to $\pi_0$ by covered arrow reversals s.t. $\cG_{\pi_1}$ has fewer edges than $\cG_{\pi_0}$. This follows readily from the existence of the Chickering sequence: we may take the highest index DAG in the sequence that is not a minimal IMAP of $\cG$. Such a DAG is guaranteed to exist: since $|\cG_{\pi_0}| > |\cG|$, there is at least one edge addition in the Chickering sequence.

In this section, we consider adding background knowledge of the following forms:
\begin{enumerate}
    \item \textbf{Known Adjacencies}: $i$ is adjacent to node $j$.
    \item \textbf{Known Order Information}: For the partition $U, V$ of $[p]$, $U <_{\pi^*} V$, i.e. if $i \in U$ and $j \in V$, then $j$ is not an ancestor of $i$ in $\cG^*$.
\end{enumerate}

\begin{algorithm}[!b]
	\caption{GSP + Background}
	\label{alg:gsp-background}
	\begin{algorithmic}
		\State{\textbf{Input:} Distribution $f$, starting permutation $\pi_0$, set of adjacent pairs $A$, partition $U, V$}
		\State{\textbf{Output:} A permutation $\pi$ and associated minimal I-MAP $\cG^\bg_{\pi}$.}
		\State 1. Set $\pi := \pi_0$;
		\State 2. Using a depth-first search with root $\pi$, search for a permutation $\tau$ such that $S_\bg(\tau) < S_\bg(\pi)$ and that the corresponding minimal I-MAP $\cG^\bg_\tau$ is connected to $\cG^\bg_\pi$ by a list of $\cI$-covered arrow reversals. If such $\tau$ exists, set $\pi$ as $\tau$ and continue this step; otherwise, return $\pi$, $\cG^\bg_\pi$.
	\end{algorithmic}
\end{algorithm}

Suppose we are given this background knowledge in the form $A = \{ (i, j) \mid i \sim j \in \cG^* \}$ and two sets $U$ and $V$. It is easy to adapt GSP so that its output satisfies these constraints. First, we define 
\begin{align*}
\cG^\bg_\pi = \{ i \rightarrow j \mid &i <_\pi j ~\textand~ 
\\
&(i \not\independent j \mid \pre_\pi(\{i, j\}) ~\textor~ (i, j) \in A) \}
\end{align*}

Then, we define $S_\bg(\pi) = \infty$ if $j <_\pi i$ for $i \in U$, $j \in V$, and $S_\bg(\pi) = |G^\bg_\pi|$ otherwise, and use this score in place of $S$. The modified algorithm is described in Algorithm~\ref{alg:gsp-background}.

Now we show that these adaptations retain consistency of GSP. The case of known adjacencies are simple. 

Since any IMAP $\cH$ of $\cG^*$ satisfies $\skel(\cG^*) \subseteq \skel(\cH)$ \citep{raskutti2018learning}, no edge $i - j$ in $K$ is deleted over the course of GSP, i.e. not testing a CI statement between $i$ and $j$ does not change the result.

No we consider the case of known order information. We show that if $\cG^*$ and $\cH$ both satisfy the known order information, then any DAG in a Chickering sequence from $\cG$ to $\cH$ will satisfy the known order information, which extends to the sequence of minimal IMAPs from some starting $\cG_\pi$ to $\cG^*$ given in the previous section.

If $\cG_m$ in the Chickering sequence satisfies the order information, there are only two scenarios in which $\cG_{m+1}$ will not: an edge $i \rightarrow j$ with $i \in A$ and $j \in B$ is reversed, or an edge is added from $j \in B$ to $i \in A$. We will show that neither scenario happens. The \textsc{ApplyEdgeOperation} algorithm only reverses edges to be in the same direction as they are in $\cH$, so the first situation never happens. The only case in which \textsc{ApplyEdgeOperation} adds an edge that is opposite its orientation in $\cH$ is in Step 8: $Y$ is a sink in $\cH$, $Z$ is a child of $Y$ in $\cG_m$, and $Y \rightarrow Z$ is not covered in $\cG_m$ because of a parent $X$ of $Y$ in $\cG_m$ that is not a parent of $\cH$ in $\cG_m$. $\cG_{m+1}$ violates the known order information only if $Z \in A$ and $X \in B$. We have two cases: if $Y \in A$, then $X \in A$ by the assumption that $\cG$ satisfies the order information. If $Y \in B$, then $Z \in B$ by the assumption that $\cG$ satisfies the known order information. Thus, in both cases, $\cG_{m+1}$ still satisfies the known order information.

\section{JCI-GSP}\label{app:jci-gsp}
For completeness, we outline JCI-GSP in Algorithm \ref{alg:jci-gsp}. Line 1 introduces variables $\zeta_k$ for each interventional setting $k \in [K]$ and lifts the distribution over $X$ to a distribution over $X$ and $\zeta$. Line 2 combines these distributions into a mixture distribution, the choice  of the uniform distribution is arbitrary for the population case. With finite samples, the weights may be picked according to the number of samples from each observational/interventional setting. Line 3 forms a permutation over both $\zeta$ and $X$ by pre-pending $\pi_0$ with an arbitrary order of the $\zeta$ variables. Line 4 encodes the \textit{generic context} background knowledge, and Line 5 encodes the background knowledge about known intervention targets. Finally, Line 7 calls GSP with the appropriate background knowledge.

\begin{algorithm}[tb]
	\caption{JCI-GSP}
	\label{alg:jci-gsp}
	\begin{algorithmic}
		\State{\textbf{Input:} Distributions $f^\obs, f^1, \cdots, f^K$ and partially known intervention sets $\cI^\kn := \{I_\kn^1, I_\kn^2 \cdots, I_\kn^K\}$, a starting permutation $\pi_0$.}
		\State{\textbf{Output:} A permutation $\pi$ and associated minimal I-MAP $\cG_{\pi}$, a complete set of estimated intervention targets $\cI := \{I^1_\pi, \cdots, I^K_\pi \}$.}
		\State 1. Let $g^\obs = \kron_{\zeta = 0} \otimes f^\obs$, $g^k = \kron_{\zeta_k = 1, \zeta_{\neg k} = 0} \otimes f^k$ for $k \in [K]$.
		\State 2. Let $g = \frac{1}{K+1} (g^\obs + \sum_{k \in [K]} g^k)$
		\State 3. Set $\pi_0' = \langle \zeta_1, \ldots, \zeta_K \rangle \cdot \pi_0$
		\State 4. Set $A_{\textrm{interventions}} = \{ \zeta_i \sim \zeta_j \mid i, j \in [K], i \neq j \}$
		\State 5. Set $A_{\textrm{targets}} = \cup_{k \in [K]} \{ \zeta_k \sim i \mid i \in I_\kn^k \}$
		\State 6. Set $A = A_{\textrm{interventions}} \cup A_{\textrm{targets}}$
		\State 7. Run GSP + Background with distribution $g$, starting permutation $\pi_0'$, adjacent pairs $A$, and partition $(\zeta, X)$.
	\end{algorithmic}
\end{algorithm}

\begin{figure}[!b]
	\centering
    \includegraphics[width=.5\textwidth]{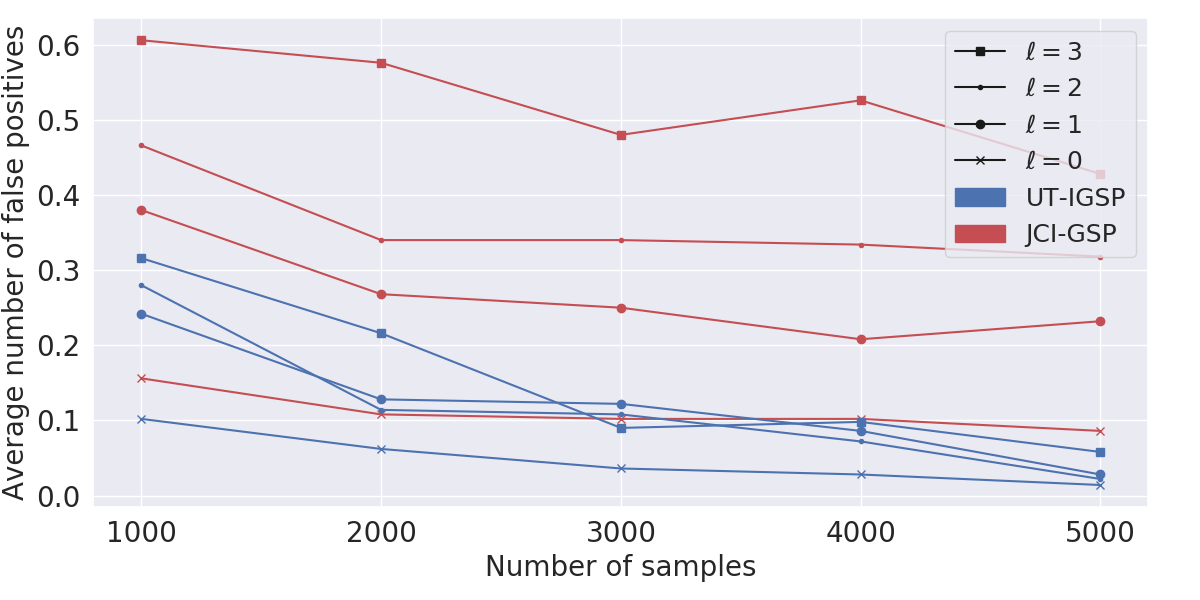}
	\caption{Performance of UT-IGSP and JCI-GSP at the task of intervention target recovery as a function of number of samples and number of off-target effects ($\ell$).} 
	\label{fig:target-recovery}
\end{figure}

\begin{figure*}[!h]
	\centering
	\begin{subfigure}[!t]{\textwidth}
	    \includegraphics[width=\textwidth]{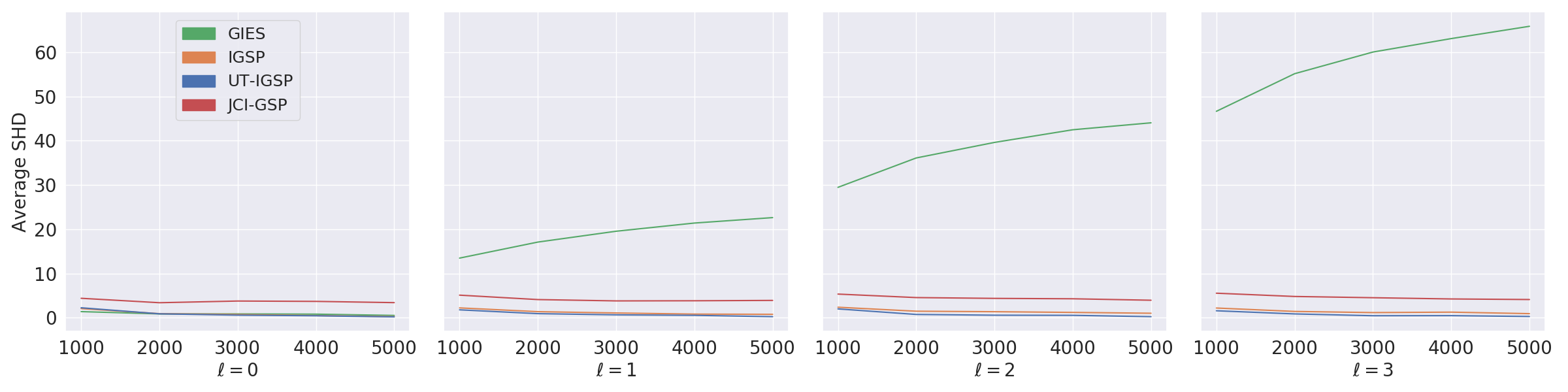}
	    \caption{Average structural Hamming distance from the true $\cI$-essential graph}
	\end{subfigure}
	\begin{subfigure}[t]{\textwidth}
	    \includegraphics[width=\textwidth]{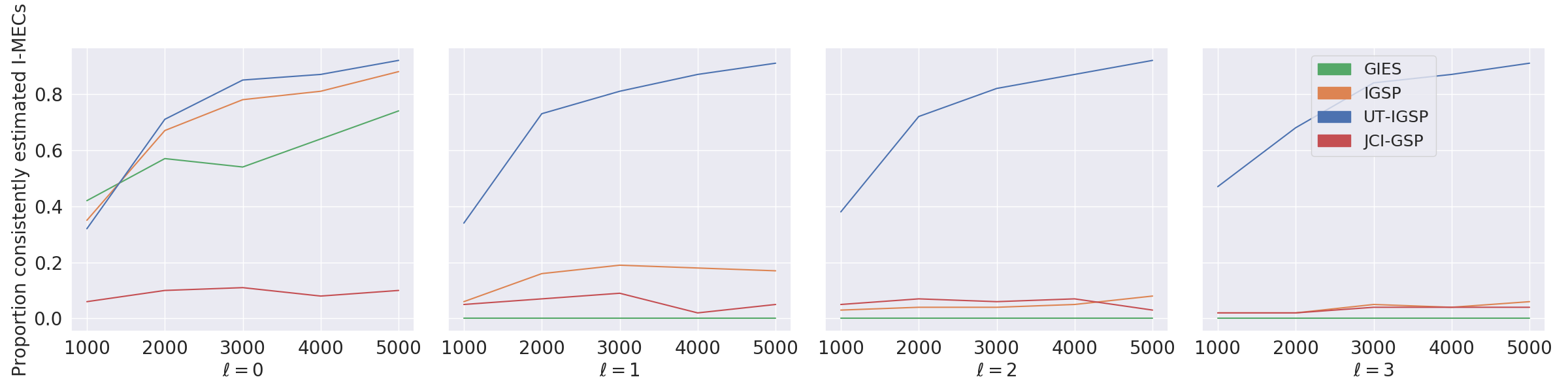}
	    \caption{Proportion of correctly estimated $\cI$-MECs}
	\end{subfigure}
	\caption{Performance of different methods as a function of number of samples and number of off-target effects ($\ell$) for 100 Gaussian DAG models on \nnodes~nodes. (a) corresponds the average Hamming distance between the estimated $\cI$-essential graph and the true $\cI$-essential graph, (b) corresponds to the proportions of consistently estimated $\cI$-MECs within the 100 randomly generated Gaussian DAG models.} \label{fig:simulations-perfect}
\end{figure*}
%

\section{Additional Evaluation}\label{appendix:additional-evaluation}

\subsection{Intervention Recovery}\label{appendix:intervention-recovery}
In Fig. \ref{fig:target-recovery}, we use the same data generated in Section \ref{subsec:simulated-data}, and report the number of false positive intervention targets for UT-IGSP. The average number of false negatives was negligble ($< .04$) for both methods.

\subsection{Perfect Interventions}\label{appendix:perfect-intervention-results}
In this section, we sample Gaussian DAG models and intervention targets in the same manner as described in Section \ref{subsec:simulated-data}. However, instead of using shift interventions, we use \textit{perfect interventions}. In particular, for $i \in I^k$, we completely remove the dependency between an intervened node and its parents (i.e., set $B_{ji} = 0$ for $j \in \pa_{\cG^*}(i)$), and change its internal noise variance to $\epsilon_i \sim \cN(1, 0.1)$. GIES was designed specifically for learning from perfect interventions, making perfect interventions a more fair comparison than shift interventions. Indeed, GIES performs better when $\ell = 0$ than it did for shift interventions, even outperforming UT-IGSP when $n = 1,000$. However, the overall trends remain: UT-IGSP outperforms GIES when the number of samples becomes large, and the performance of GIES is drastically reduced by even a single off-target intervention.

\end{document}